\documentclass[final]{siamltex}

\usepackage{amsfonts,amsmath,amssymb}
\usepackage{mathrsfs,mathtools}
\usepackage{enumerate}
\usepackage{hyperref}
\usepackage{esint}
\usepackage{graphicx}
\usepackage{subfigure}
\usepackage{algorithm}
\usepackage{algpseudocode}
\usepackage[normalem]{ulem}
\usepackage{comment}

\includecomment{ieee} 

\DeclareMathAlphabet{\mathpzc}{OT1}{pzc}{m}{it}

\newcommand{\bigo}[1]{{\cal O}({#1})}

\usepackage[usenames,dvipsnames]{color}

\newtheorem{rem}[theorem]{Remark}
\numberwithin{equation}{section}

\title{A Note on QR-Based Model Reduction: Algorithm, Software, and Gravitational Wave Applications\thanks{HA has been supported in part by NSF grant DMS-1521590. SEF has been supported in part by NSF grants PHY-1606654 and the Sherman Fairchild Foundation.}}

\author{Harbir Antil\thanks{Department of Mathematical Sciences, George Mason University, Fairfax, VA 22030, USA. \texttt{hantil@gmu.edu}},
\and 
Dangxing Chen\thanks{Department of Mathematics, University of North Carolina. Chapel Hill, NC 27514, USA}. ({\tt dangxing@live.unc.edu})
\and 
Scott E. Field\thanks{Department of Mathematics. University of Massachusetts, Dartmouth, MA 02747, USA. ({\tt sfield@umassd.edu})}}

\date{Draft version of \today.}

\begin{document}

\maketitle
\begin{abstract}
While the proper orthogonal decomposition (POD) is optimal under certain norms it's also expensive to compute. For large matrix sizes, it is well known that the QR decomposition provides a tractable alternative. Under the assumption that it is rank--revealing QR (RRQR), the approximation error incurred is similar to the POD error and, furthermore, we show the existence of an RRQR with exactly same error estimate as POD. To numerically realize an RRQR decomposition, we will discuss the (iterative) modified Gram Schmidt with pivoting (MGS) and reduced basis method by employing a greedy strategy. We show that these two, seemingly different approaches from linear algebra and approximation theory communities are in fact equivalent. Finally, we describe an MPI/OpenMP parallel code that implements one of the QR-based model reduction algorithms we analyze. This code was developed with model reduction in mind, and includes functionality for tasks that go beyond what is required for standard QR decompositions. We document the code's scalability and show it to be capable of tackling large problems. In particular, we apply our code to a model reduction problem motivated by gravitational waves emitted from binary black hole mergers and demonstrate excellent weak scalability on the supercomputer Blue Waters up to $32,768$ cores and for complex, dense matrices as large as $10,000$-by-$3,276,800$ (about half a terabyte in size).

\end{abstract}

\begin{keywords}
greedy algorithm, QR decomposition, rank revealing, low-rank approximations, software
\end{keywords}
\begin{AMS}

\end{AMS}
%

\section{Introduction} \label{s:intro}

Algorithms to compute low-rank matrix approximations have enabled many recent scientific and engineering advances. In this CiSE special issue, we summarize the theoretical properties of the most influential low-rank techniques. We also show two of the most popular techniques are algorithmically equivalent and describe a massively parallel code for QR-based model reduction that has been used for gravitational wave applications. 
This preprint is an expanded, more technical version of the manuscript published in IEEE's Computing in Science \& Engineering. 

In this paper we consider both practical and theoretical low-rank approximations found by singular value decomposition (SVD) or QR decomposition of a matrix $S$, presenting error estimates, algorithms and properties of each. Both decompositions can be used, for example, to compute a low-rank approximation to a matrix (a common task in numerical linear algebra) or provide a high-fidelity approximation space suitable for model order reduction (a common task in engineering or approximation theory).

For certain norms an SVD-based approximation is optimal. However, for many large problems the (classical) SVD becomes problematic in terms of its memory footprint, FLOP count, and scalability on many-core machines. By comparison, QR-based model reduction is computationally competitive; it carries a lower FLOP count, is easily parallelized, and has a small inter-process communication overhead, thereby allowing one to efficiently utilize many-core machines. Indeed, for large matrices some SVD algorithms are based on QR decompositions~\cite{constantine2014model}. Furthermore, for certain matrices $S$, we will show that the SVD and a special class of QR decompositions share similar approximation properties. 

We are especially interested in the setting where the snapshot (or ``data") matrix may be too large to load into memory thereby precluding straightforward use of the singular value or, equivalently, a proper orthogonal decomposition (POD). In order to salvage an SVD approach, randomized or hierarchical methods can be used. QR factorizations have long been recognized
as an alternative low-rank approximation. For instance the \emph{rank revealing} QR (RRQR) factorization \cite{MR1157596, chan1987,MR1266606,MR0181094} computes a decomposition of a matrix $S \in \mathbb{R}^{N \times M}$ as 
\begin{equation} \label{eq:qr1}
S\Pi = QR = Q 
    \left[ \begin{array}{cc}
           R_{11} & R_{12} \\ 
           0 & R_{22} \end{array} \right] ,
\end{equation}
where $Q \in \mathbb{R}^{N \times N}$ is orthogonal, $R_{11} \in \mathbb{R}^{k\times k}$ is upper triangular, $R_{12} \in \mathbb{R}^{k \times (M-k)}$, and $R_{22} \in \mathbb{R}^{(N-k) \times (M-k)}$. The column permutation matrix $\Pi$ is usually chosen such that $\| R_{22} \|_2$ is small and $R_{11}$ is well-conditioned. This factorization~\eqref{eq:qr1} was introduced in \cite{MR0181094}, and the first algorithm to compute it is based on the QR factorization with column pivoting~\cite{MR0176590}. We also refer to a recent work on this subject \cite{demmel2015communication}.

While an RRQR always exists (see Sec.~\ref{s:orrqr}), it may be computationally challenging to find. We shall consider two specific QR strategies: modified Gram Schmidt (MGS) and a reduced basis method using a greedy strategy (RB--greedy). Although the former algorithm is widely known within the linear algebra community, the latter has become extremely popular in the approximation and numerical analysis communities \cite{PBinev_ACohen_RDevore_2010a,Devore2012,buffa2012priori}. We show that finite dimensional versions of these two approaches produce equivalent basis sets and discuss their error estimates. While for a generic $S$ these algorithms may not provide a RRQR, in all practical settings with which we are familiar these algorithms are rank revealing and the resulting RRQR approximation error is of the same order as the SVD/POD. There may be additional advantages when the columns form the basis as opposed to linear combinations over all columns; a typical example is column subset selection~\cite{tropp2009column}.

As a rank-revealer, the column pivoted QR decomposition is known to fail on, for example, Kahan's matrix~\cite{hong1992}. A formal fix to this is discussed in \cite[Section~4]{demmel2015communication}, see also \cite{MR2528516} where several related issues were analyzed and the appropriate algorithmic fixes were discussed. Nevertheless, matricies like the Kahan one are rarely (if ever) encountered in model reduction problems. In typical cases, the approximation properties of QR-based model reduction is summarized as follows. The RB-greedy error in Algorithm~\ref{algo:RB-greedy} is given by $\max_{1 \le i \le M} \| s_i - Q_k Q_k^T s_i \|_2$ where $s_i \in \mathbb{R}^N$ are columns of $S$ and $Q_k \in \mathbb{R}^{N\times k}$ (see Definition~\ref{fqr}). The state-of-the-art results presented in \cite{Devore2012, PBinev_ACohen_RDevore_2010a} provide us an a priori behavior of this error: if the Kolmogorov $n$-width (best approximation error) decays exponentially with respect to $k$ so does the greedy error. For many model reduction problems, smoothness with respect to parametric variation plays an essential role. For smooth models
the $n$-width (and thus the greedy error) is expected to decay exponentially 
fast~\cite{Devore2012,APinkus_1985}. 

We will show that MGS is equivalent to RB--greedy (see Proposition~\ref{prop:MGS_RBgreedy}) and derive error estimates for both algorithms. We recall error estimates for the \emph{full} QR decomposition in Theorems~\ref{thm:RBminimize}--\ref{theorem:qrmax} and, under the assumption that this decomposition is an RRQR, we show that the underlying error is of same order as POD in the $\ell^2$--norm. Existence of an \emph{optimal} RRQR decomposition is shown. 
\begin{ieee}
We give a reconstruction strategy in Section~\ref{s:recon}, which is cheaper than, but as accurate as, the POD. 
\end{ieee}

A key contribution of this paper is the development of a
publicly available code~\footnote{The code is available at \url{https://bitbucket.org/sfield83/greedycpp/}.} that implements the RB-greedy algorithm parallelized with message passing interface (MPI) and OpenMP. Unlike other parallelized QR codes, our software is designed with model reduction in mind and uses a simple interface for easy integration with model-generation codes. Sec.~\ref{s:numerics} documents the code's performance for dense matrices with sizes as large as $10^7$-by-$10^4$. Model reduction is sometimes combined with an empirical interpolation method, and we briefly document our codes efficiency in computing empirical interpolants~\cite{Maday_2009,chaturantabut:2737} using many thousands of basis. We focus on generating empirical interpolants for the acceleration of gravitational wave parameter 
inference~\cite{HAntil_SField_RHNochetto_MTiglio_2013,abbott2017first,smith2016fast,canizares2015accelerated,meidam2018parametrized}; the QR-accelerated inference codes have been used in the most recent
set of gravitational wave detections~\cite{abbott2017first,abbott2017gw170817,abbott2017gw170814}. 
For such large dimensional reduction problems, an efficient, parallelized 
code~\cite{greedycpp} 
running on thousands of cores has proven essential. 

\begin{ieee}
The outline of this paper is as follows. In Section~\ref{s:prom} we introduce projection based reduced order model (ROM) techniques. We summarize well known facts about POD/SVD-based model reduction in Section~\ref{s:POD} such as optimality results and error bounds. Section~\ref{s:fqr} discusses the \emph{full} QR factorization and the resulting approximation. Section~\ref{s:rrqr} motivates rank revealing QR-based model reduction as a computationally efficient alternative and provides error bounds and comparisons to POD. Two specific QR-based algorithms (MGS and RB--greedy) are considered and compared in Section~\ref{s:eqmgsrb}, and reconstruction technique is presented in Section~\ref{s:recon}. Section~\ref{s:numerics} documents performance and scalability tests of the open-source greedycpp code developed in this paper~\cite{greedycpp}. 
\end{ieee}

\section{Dimensional reduction techniques} \label{s:prom}
Let us assume we are given $M$ samples $ s_1, ..., s_M \in \mathbb{R}^N$ and an associated snapshot matrix $S=(s_1, ..., s_M) \in \mathbb{R}^{ N \times M}$ whose $i^\mathrm{th}$ column is $s_i$. Each $s_i$ corresponds to a realization of an underlying parameterized model: we evaluate the model at selected parameter values and designate the solution as $s_i$. 

Within the setting just described, reduced order models are derived from a low-rank approximation for $S$. As briefly summarized in this section, the SVD and QR exposes certain kinds of low-rank approximations. We introduce a few definitions.

\begin{definition}[full SVD] \label{fsvd}
Given a matrix $S \in \mathbb{R}^{N \times M}$, the \emph{full SVD} of $S$ is
\[
    S = V \Sigma W^T,
\]
where $V \in \mathbb{R}^{N \times N}, \Sigma \in \mathbb{R}^{N \times M}, W \in \mathbb{R}^{M \times M}$. In addition, $V$ and $W$ are orthogonal matrices, and $\Sigma$ is a diagonal matrix with non-increasing entries, known as singular values. The $k^{\tt th}$ singular value is denoted by $\sigma_k$. 
\end{definition}

From the singular values we define the \emph{ordinary} and \emph{numerical-ranks} of a matrix as follows:
\begin{definition}[ordinary- and numerical-ranks of $S$] \label{def:numerical_rank}
Let $S \in \mathbb{R}^{N \times M}$ be a matrix whose singular values $\{ \sigma_i \}_{i=1}^{M}$ are arranged in a decreasing order. Then $S$ is said to have numerical rank $k$ if
\begin{align*}
    \sigma_{k+1} \approx \epsilon_{\tt mach}
\end{align*}
where $\epsilon_{\tt mach}$ is the machine precision and a standard, or ordinary-rank, if
\begin{align*}
    \sigma_{k+1} = 0 \, .
\end{align*}
\end{definition}

\begin{definition}[full QR] \label{fqr}
The \emph{full $QR$} factorization of $S \in \mathbb{R}^{N \times M}$ is
\begin{align}\label{eq:qr}
    S\Pi = QR = \left[ \begin{array}{cc} 
    Q_k & Q_{N-k} \end{array} \right] 
    \left[ \begin{array}{cc}
           R_{11} & R_{12} \\ 
           0 & R_{22} \end{array} \right] . 
\end{align}
where $Q_k \in \mathbb{R}^{N \times k}$ and $Q_{N-k} \in \mathbb{R}^{N \times (N-k)}$ are orthogonal, $R_{11} \in \mathbb{R}^{k \times k}$ is upper triangular, $R_{12} \in \mathbb{R}^{k \times (M-k)}$, and $R_{22} \in \mathbb{R}^{(N-k) \times (M-k)}$. 
\end{definition}

The role of a permutation matrix $\Pi \in \mathbb{R}^{M \times M}$ in \eqref{eq:qr} is to swap columns of $S$ and is crucial for achieving QR-based model reduction. Different QR algorithms prescribe different rules for discovering $\Pi$.
If QR is not pivoted, then we define the permutation matrix as identity. 

Of particular interest is the RRQR decomposition. There are several different ways of defining an RRQR, one of them \cite{hong1992, MR1266606} says that the factorization \eqref{eq:qr} is an RRQR if:
\begin{definition}[RRQR] \label{def:RRQR} 
Assume $S \in \mathbb{R}^{N \times M}$ has numerical rank $k$, if
\begin{align*}
\sigma_{\tt min} (R_{11}) \gg \| R_{22} \|_2 \approx \epsilon_{\tt mach}
\end{align*}
then the factorization $S \Pi = QR$ is called a \em{Rank Revealing QR factorization} (RRQR) of $S$. 
\end{definition}

From~\cite[Lemma 1.2]{hong1992} we recall that the following holds for any $\Pi$
\[
 \sigma_k(S) \ge \sigma_{\rm min}(R_{11}) \quad 
 \mbox{and} \quad 
 \| R_{22} \|_2 = \sigma_{\rm max}(R_{22}) \ge \sigma_{k+1}(S) , 
\]
whence Definition~\ref{def:RRQR} implies
\[
 \sigma_{k+1}(S) \le \| R_{22} \|_2 = \sigma_{\rm max}(R_{22}) \ll \sigma_{\rm min}(R_{11}) \le \sigma_{k}(S) ,
\]
i.e., Definition~\ref{def:RRQR} introduces a large gap between $\sigma_{k+1}(S)$ and $\sigma_{k}(S)$. Finally, we define an {\em optimal RRQR} of $S$ as follows:
\begin{definition}[optimal RRQR] \label{def:RRQR_optimal} 
$QR=S \Pi$ is an optimal RRQR of $S \in \mathbb{R}^{N\times M}$ if
\begin{align}
\| S-Q_k Q_k^TS\|_2 = \sigma_{k+1}.
\end{align}
\end{definition}

Projection-based model reduction represents a single column of the matrix, $s_i$, via orthogonal projection of $s_i$ onto the span of the basis. Approximation errors using an SVD basis, $s_i - V_k V_k^T s_i$, or using a QR basis, $s_i - Q_k Q_k^Ts_i$, are considered in the next sections. Throughout this paper we will use $V_k$ to denote a matrix formed by the first $k$ columns of $V$.

\section{Full POD/SVD with error estimates} \label{s:POD}

We first recall the POD problem formulation: A POD computes $k$ orthonormal vectors $ v_1, ..., v_k \in \mathbb{R}^N$ which provide an optimal solution to 
\begin{equation} \label{eq:min_POD}
\begin{aligned}
    \epsilon_* := \min_{Y_k \in \mathbb{R}^{N \times k} }   \|S - Y_k Y_k^T S\|_*^2 
\end{aligned}
\end{equation}
where $*$ is either the Frobenius $(F)$ or matrix 2-norm. It is well known that the $*$--norm solution to \eqref{eq:min_POD} can be computed by first performing a SVD of $S = V \Sigma W^T$, from which $Y_k = V_k = (v_1, ..., v_k) \in \mathbb{R}^{ N \times k}$ is simply the first $k$ columns of $V$. 

We shall assume, for definiteness, a \emph{full SVD} of $S$. We will frequently require matrices with zeros in all columns after $k+1$ and shall denote these with a superscript ``$0$". For example, $V_k V_k^T V := V_k^0$ is $V$ with zeros in columns $k+1$ to $N$. Using this notation 
it is easy to see that the POD approximation of $S$
\begin{align} \label{eq:SVD_POD}
V_k V_k^T S = V_k V_k^T V \Sigma W^T = V_k^0 \Sigma W^T = V \Sigma_k^0 W^T = \sum_{i=1}^k \sigma_i v_i w_i^T 
\end{align}
is exactly a sum of $k$ rank-one matrices, 
where $\Sigma_k^0$ is $\Sigma$ with zeros in columns $k+1$ to $N$. This illustrates the close connection of POD with the \emph{partial SVD} factorization $V \Sigma_k^0 W^T$.

\subsection{2-norm and $F$-norm POD error estimates}\label{s:2FfPOD}

We briefly recall standard POD error estimates. 
\begin{ieee}
\begin{lemma} \label{lem:frobenius_norm2}
Let $A$ and $B$ be two matrices. If $ A $ is orthogonal then $ \| AB \|_* = \| B \|_* $; if $ B $ is orthogonal, then $  \| AB \|_* = \| A \|_*$.     
\end{lemma}
\end{ieee}

\begin{theorem}[POD error estimates] \label{thm:POD}
Given $S \in \mathbb{R}^{N\times M}$ with SVD of $S = V \Sigma W^T$, then
\begin{enumerate}
    \item[$(i)$] $\|S-V_k V_k^T S\|_F^2 = \sum_{j=k+1}^{\min{ \{  M,N \} }} \sigma_j^2$.  

    \item[$(ii)$] $\|S-V_k V_k^T S\|_2 = \sigma_{k+1}$.  
\end{enumerate}
\end{theorem}
\begin{ieee}
\begin{proof} 
The proof is standard but we recall it here for completeness. By Lemma~\ref{lem:frobenius_norm2} and Equation~\eqref{eq:SVD_POD}
\begin{align*}
    \|S-V_k V_k^T S\|_*^2 & = \| V \Sigma - V_k V_k^T V \Sigma \|_*^2 = \| \Sigma - \Sigma_k^0 \|_*^2  \, .
\end{align*}
From the definition of $F$ and 2 norms, we deduce
\[
    \| \Sigma - \Sigma_k^0 \|_F^2   = \sum_{j=k+1}^{\min{ \{  m,N \} }} \sigma_j^2,
        \qquad 
    \| \Sigma - \Sigma_k^0 \|_2   = \sigma_{k+1} \, ,
\]
which completes the proof. 
\hfill
\end{proof}
\end{ieee}

\begin{ieee}
From Theorem~\ref{thm:POD}, errors measured in the $F$-norm require computation of all singular values, which can be expensive. Using the 2-norm requires only the first $k+1$ singular values, which motivates a choice of $\sigma_{k+1} < \tau$ to control the approximation error $\tau$ in Algo.~\ref{algo:POD}. In practice, we always choose $\tau$ larger than machine precision.

\hspace{0.5cm}

{\scriptsize
\begin{algorithm}[H]
\caption{POD}
\label{algo:POD}
\begin{algorithmic}[1]
\State \textbf{Input}: Snapshot matrix $S = (s_1, ... , s_M) \in \mathbb{R}^{N \times M}$ and tolerance $\tau >0$
\State \textbf{Output}: $ V_k = (v_1, ... , v_k) \in \mathbb{R}^{N \times k}$ 
\vskip 10pt
\State Compute the singular value decomposition $S = V \Sigma W^T $.
\State Find smallest index $ k $ such that the singular values satisfy $ \sigma_{k+1} < \tau.$ 
\State Return the first $ k $ columns $V_k=(v_1, ... , v_k) \in \mathbb{R}^{N \times k} $ of $V$.
\end{algorithmic}
\end{algorithm}
}
\end{ieee}

\section{Full QR with error estimates} \label{s:fqr}

In Theorem~\ref{thm:POD} we showed that POD provides the best rank $k$ $*$-norm approximation to the snapshot matrix $S$. In Theorem~\ref{thm:RBminimize} we will see that under the assumption that the decomposition is an RRQR (according to Def. ~\ref{def:RRQR}) the resulting approximation error using \emph{full} QR factorization is of the same order as POD-based approximation error. In Section~\ref{s:2Ffqr} we will first discuss the $*$--norm error estimates and we conclude with max-norm error estimates in Section~\ref{s:mfqr}. 
Note that a QR decomposition always exists but need not be unique~\cite{trefethen1997numerical}.

\subsection{2-norm and $F$-norm error estimates}\label{s:2Ffqr}

We define the QR-based error for the decomposition \eqref{eq:qr} as
\begin{align}\label{eq:qr_err_star} 
    \epsilon_*^{QR} (\Pi) := \|S\Pi-Q_kQ_k^T S \Pi \|_* = \| S - Q_kQ_k^T S\|_*,
\end{align}
where $*$ could be either 2-norm or $F$-norm. The QR-based approximation error 
depends on the permutation matrix implicitly through $Q_k$. 
To avoid extra notation, we sometimes omit writing $S\Pi$ and assume $S$ is already pivoted when it is clear from context. 

Let
\begin{align} \label{eq:R_svd}
    R = \overline{V} \Sigma \overline{W}^T ,
\end{align}
be the SVD of $ R $ where we have added over-bars to make a clear distinction with the SVD of $S$ and, recall, 
that $S$ and $R$ have the same singular value spectrum (hence $\overline{\Sigma} = \Sigma$).
Furthermore $\overline{V} \in \mathbb{R}^{N \times N}$, $\overline{W} \in \mathbb{R}^{M \times M}$ and $\Sigma \in \mathbb{R}^{N \times M}$. Let 
\begin{align} \label{eq:Qk0}
Q_k^0 := Q_k Q_k^T Q \, 
\end{align}
 be $Q$ with zeros in columns from $k+1$ to $N$, $ \overline{V}_{k,0}$ is $ \overline{V} $ with zeros from rows  $ k+1 $ to $N$, and $\overline{V}_{N-k,0}$ is $\overline{V}$ with zeros from rows $1$ to $k$. Observe that
\begin{align} \label{eq:Vk0}
Q_k^0 \overline{V} = Q \overline{V}_{k,0}.
\end{align}

\begin{theorem}[full QR $*$-norm error estimate] \label{thm:RBminimize} 
Let $Q_k \in \mathbb{R}^{N\times k}$ denote the first $k$ columns of $Q$, then
\begin{align}\label{eq:RBminimize}
    \epsilon_*^{QR} (\Pi) = \| S - Q_kQ_k^TS \|_* =\| \overline{V}_{N-k,0} \Sigma\|_* = \|R_{22} \|_*,
\end{align}
where $*$ is either 2-norm or F-norm. A proof can be found in Ref.~\cite{chan1987}
\end{theorem}

\begin{ieee}
\begin{proof}
To show the first equality, we use Lemma~\ref{lem:frobenius_norm2}, ~\eqref{eq:R_svd}, and ~\eqref{eq:Qk0}, we deduce
\begin{align*} 
    \| S - Q_kQ_k^TS \|_*  & = \| QR - Q_k Q_k^T Q R \|_* = \| QR - Q_k^0 R\|_* \\
      & = \| Q \overline{V} \Sigma \overline{W}^T - Q_k^0 \overline{V} \Sigma \overline{W}^T \|_* = \| Q \overline{V} \Sigma - Q_k^0 \overline{V} \Sigma \|_* .
\end{align*}  
Invoking ~\eqref{eq:Vk0}, we obtain
\begin{align*}
\| S - Q_kQ_k^TS \|_*
  &= \| Q \overline{V} \Sigma - Q \overline{V}_{k,0} \Sigma \|_*
   = \| (\overline{V} - \overline{V}_{k,0} ) \Sigma\|_*
   = \| \overline{V}_{N-k,0} \Sigma \|_* \, ,
\end{align*}
which proves the first equality in \eqref{eq:RBminimize}.

The proof of the second equality is similar to the proof of Theorem~\ref{thm:POD}, we readily obtain
\[
\| S - Q_kQ_k^T S \|_* = \| QR-Q_k^0 R \|_* = \|R_{22} \|_*,
\]
the last equality naturally comes out as $R$ upper triangular. Thus, we conclude. 
\hfill
\end{proof}
\end{ieee}

\begin{rem}\rm
Invoking the optimality of POD, we notice that $\| S - V_k V_k^T S \|_{*} \le \| S - Q_k Q_k^T S \|_{*}$. Assuming that QR is RRQR (according to Def.~\ref{def:RRQR}) and employing Theorem~\ref{thm:RBminimize} we notice that the approximation error of the \emph{full} QR and POD are of the same order. 
\end{rem}

\subsection{Max--norm error estimate}\label{s:mfqr}

We define the QR--based approximation error in the max-norm as
\begin{align} \label{eq:qr_err_max}
    \epsilon_{\max}^{QR} (\Pi) := \max_{1 \leq i \leq M} \| s_i - Q_k Q_k^T s_i \|_2 \, , 
\end{align} 
where, for vectors, $\| \cdot \|_2$ is the usual Euclidean norm. Measured in the max-norm, the QR-based approximation error is given by the following result.

\begin{theorem}[full QR max-norm error estimate] \label{theorem:qrmax} 
Let $r_i$ be the $i$-th column of $R$ and $\widetilde{r}_i$ be a subvector of $r_i$ from row $k+1$ to $N$. Then the approximation error for any given $s_i$ is 
\[
 \| s_i - Q_k Q_k^T s_i \|_2 =  \| \widetilde{r}_i \|_2, \ i = 1, \dots, M , 
\]
whence $\epsilon_{max}^{QR} = \max_{1 \leq i \leq M} \| \widetilde{r}_i \|_2$.
\end{theorem}

\begin{ieee}
\begin{proof}
Invoking Lemma~\ref{lem:frobenius_norm2} and Equation~\eqref{eq:Qk0}, we deduce 
\[
\| s_i- Q_k Q_k^T s_i \|_2 = \| Qr_i - Q_k Q_k^T Qr_i \|_2 = \| Qr_i - Q_k^0 r_i \|_2 = \| \widetilde{r}_i \|_2 \, ,
\]
thus we conclude. 
\end{proof}
\end{ieee}

\begin{corollary}\label{cor:emax_es}
The following relation between the $*$ and max norm estimates for QR holds 
\[
\epsilon_{max}^{QR} \le \epsilon_*^{QR}.
\]
\end{corollary}

\begin{ieee}
\begin{proof}
Invoking Theorem~\ref{thm:RBminimize} and Theorem~\ref{theorem:qrmax} and the fact that $ \widetilde{r}_i$ is the submatrix of $R_{22}$, we conclude
\[
\max_{1 \leq  i \leq M} \|\widetilde{r}_i \|_* \leq \| R_{22} \|_*,
\]
which is due to the fact that the norm of a submatrix is less than the norm of matrix it embedded in. 
\end{proof}
\end{ieee}

\section{Specific QR algorithms} \label{s:rrqr}

The goal of this section is to discuss different QR decomposition strategies (that is to say, the pivoting strategy) 
useful for low-rank approximation. We first consider the \emph{optimal} RRQR whose 2-norm error estimates exactly match the POD
error estimates (see Theorem~\ref{thm:POD}). Next we consider \emph{practical} RRQR algorithms which are implementable.
For the \emph{practical} RRQR, in Section~\ref{s:prrqr}, we will discuss two algorithms: 
modified Gram--Schmidt (MGS) with pivoting and a reduced basis method using a greedy approach (RB--greedy). 
We show that these methods are in fact equivalent in certain settings. We will derive the error estimates and furnish FLOP counts. 
\begin{ieee}
In Section~\ref{s:recon} we discuss a QR-reconstruction strategy with error estimates similar to POD.
\end{ieee}

\subsection{Optimal RRQR} \label{s:orrqr}

In this section we show that an optimal QR-based approximation exits (although its not necessarily unique). We will provide a constructive proof. We begin by computing the SVD of $S$ as 
\begin{align}\label{eq:oqr_1}
S = V \Sigma W^T = 
\left[ \begin{array}{cc}
V_k & V_{N-k}
\end{array} \right]
\left[ \begin{array}{cc}
\Sigma_k & 0 \\
0 & \Sigma_{N-k,M-k} \end{array} \right]
\left[ \begin{array}{cc}
W_{k}^T \nonumber \\
W_{M-k}^T \end{array} \right]  \\
= V_k \Sigma_k W_k^T+V_{N-k} \Sigma_{N-k,M-k} W_{M-k}^T.
\end{align}
In addition, let 
\begin{align}\label{eq:oqr_3}
\mathcal{QR} = \Sigma_k W_k^T ,
\end{align}
be a QR factorization of $\Sigma_k W_k^T$. Finally, we set 
\begin{align}\label{eq:oqr_4}
Q_k = V_k \mathcal{Q} \, .
\end{align}
We will prove that such a $Q_k$ will lead to the \emph{optimal} RRQR according to Def.~\ref{def:RRQR}. 

\begin{theorem}[existence of optimal RRQR and equivalence to POD error] \label{thm:RRQR}
Let $S \in \mathbb{R}^{N \times M}$, then an optimal RRQR (according to Def.~\ref{def:RRQR}) of $S$ exists, with 
\begin{equation}\label{eq:oqrmat}
Q = \left[ \begin{array}{cc} 
Q_k & Q_{N-k} \end{array} \right]
\quad \mbox{and} \quad 
R = 
\left[ \begin{array}{cc}
\mathcal{R} \\
0 \end{array} \right] ,
\end{equation}
where $\mathcal{R}$ and $Q_k$ are defined in \eqref{eq:oqr_3} and \eqref{eq:oqr_4} respectively and $Q_{N-k} \in \mathbb{R}^{N \times (N-k)}$ is defined such that $Q^T Q = I$. The following estimate holds
\begin{equation}\label{eq:oqrpod}
    \| S - Q_kQ_k^TS \|_2 = \| S - V_k V_k^T S \|_2 = \sigma_{k+1}.
\end{equation}
where $V_k$ is as defined in Theorem~\ref{thm:POD}. 
\end{theorem}

\begin{ieee}
\begin{proof}
The proof is constructive. Using the SVD of $S$ from \eqref{eq:oqr_1}, it is not difficult to see that 
\begin{align}\label{eq:oqr_2}
\| S - V_k \Sigma_k W_k^T \|_2 = \sigma_{k+1} . 
\end{align}

Next we use \eqref{eq:oqr_4}, \eqref{eq:oqr_1}, and invoke the orthogonality of $V$ to yield
\begin{align*}
\| S - Q_kQ_k^TS \|_2 
&= \| S - V_k \mathcal{Q} \mathcal{Q}^T V_k^T S \|_2 \\
&= \| S - V_k \mathcal{Q} \mathcal{Q}^T V_k^T \left( V_k \Sigma_k W_k^T+V_{N-k} \Sigma_{N-k,M-k} W_{M-k}^T \right) \|_2 \\
&= \| S - V_k \mathcal{Q} \mathcal{Q}^T \Sigma_k W_k^T \|_2 \, . 
\end{align*}
Notice that $\mathcal{Q} \mathcal{Q}^T \Sigma_k W_k^T = \mathcal{Q} \mathcal{Q}^T \mathcal{QR} = \mathcal{QR} = \Sigma_k W_k^T$, where we have used \eqref{eq:oqr_1}, we arrive at
\[
\| S - Q_kQ_k^TS \|_2 = \| S - V_k \Sigma_k W_k^T \|_2 = \sigma_{k+1} , 
\]
therefore \eqref{eq:oqr_2} in conjunction with Theorem~\ref{thm:POD}  enables us to obtain the asserted equality \eqref{eq:oqrpod}. 
\end{proof}
\end{ieee}

We do not offer an efficient algorithm to calculate the \emph{optimal} RRQR other than to first calculate a potentially expensive SVD as
in proof. 
Notice that in this case the optimal RRQR's permutation matrix $\Pi$ is the identity. In particular, the QR factorization constructed in the proof does not arise from the QR pivoting strategies of Sec.~\ref{s:prrqr}, and to the best of our knowledge the matrix $Q_k$ cannot in general be constructed as a column subset of $S$.

\begin{ieee}
The estimate in \eqref{eq:oqrpod} states that 2-norm error in POD (Theorem~\ref{thm:POD}) is same as the optimal RRQR. We reemphasize that, although the proof above is constructive, it does not provide a numerical recipe to compute QR factorization. This is the subject of next section. 
\end{ieee}

\begin{corollary}[optimal RRQR and ordinary $k$--rank matrix] \label{cor:oqrkrank}
If $S \in \mathbb{R}^{N \times M}$ has an ordinary rank $k$ then under the assumptions of Theorem~\ref{thm:RRQR}, $S = QR$ with $Q$ and $R$ defined in \eqref{eq:oqrmat}. 
\end{corollary}
\begin{ieee}
\begin{proof}
Using \eqref{eq:oqr_3} and \eqref{eq:oqr_4} we first notice that 
\[
  V_k \Sigma_k W_k^T = V_k \mathcal{QR} = Q_k R = Q R . 
\]
Since $S$ has ordinary rank $k$ therefore $\sigma_{k+1} = 0$, invoking \eqref{eq:oqr_2} we obtain the assertion. 
\end{proof} 
\end{ieee}

\subsection{Practical RRQR} \label{s:prrqr}

This section is devoted to two algorithms that aim to compute an RRQR. 
Algorithm~\ref{algo:POD-QR} is the modified Gram-Schmidt (MGS) with pivoting~\cite{GHGolub_CFLoanVan_1996a}. Algorithm~\ref{algo:RB-greedy} is a particular flavor of the (by now) standard reduced basis (RB)-greedy algorithm~\cite{PBinev_ACohen_RDevore_2010a,Devore2012,buffa2012priori}. RB-greedy is a popular tool employed in the  construction of model reduction schemes for parameterized partial differential equations. We present each algorithm in their standard presentation, and in Proposition~\ref{prop:MGS_RBgreedy} show these algorithms to be equivalent when the snapshots $s_i$ are elements of an $N$--dimensional Euclidean vector space. Theorem~\ref{theorem:qrmax} and Corollary~\ref{cor:qrstop} motivate the choice of stopping criteria, which relies on diagonal entries of $R$ being non-increasing. This aspect is discussed in Corollary~\ref{cor:qrstop}. 

\subsubsection{Equivalence of the MGS and RB-greedy algorithms} \label{s:eqmgsrb}

The MGS with pivoting and RB-greedy algorithms are as follows:

\noindent\begin{minipage}{.5\textwidth}
{\scriptsize
\begin{algorithm}[H]
\caption{MGS with pivoting}
\label{algo:POD-QR} 
\begin{algorithmic}[1]
\State \textbf{Input}: $S = (s_1, ... , s_M)$, $\tau >0$
\State \textbf{Output}: $ Q_k = (q_1, ... , q_k)$
\vskip 10pt
\State Set $V = S$ and $k = 1$
\State Set $R(1,1) = \mbox{max}_{j} \| V(:,j) \|_2$
\While{$R(k,k) > \tau$}
\State $i$ = $\mbox{arg max}_{j} \| V(:,j) \|_2$
\State $\Pi(k,i)=1$
\State $R(k,k) = \| V(:,i) \|_2$
\State $Q(:,i) = V(:,i)/R(k,k)$
\For{$j = k+1$ to $M$}
\State $R(k,j) = Q(:,k)^TV(:,j)$
\State $V(:,j) = V(:,j)-R(k,j)Q(:,k)$
\EndFor
\State $k = k+1$
\EndWhile
\end{algorithmic}
\end{algorithm}
}
\end{minipage}
\begin{minipage}{.5\textwidth}
{\scriptsize
\begin{algorithm}[H]
\caption{RB-greedy}
\label{algo:RB-greedy}
\begin{algorithmic}[1]
\State \textbf{Input}:  $S = (s_1, ... , s_M)$, $\tau >0$
\State \textbf{Output}: $Q_k = (q_1, ... , q_k)$
\vskip 10pt
\State Set $ k=0 $ 
\State Define $\hat{\sigma}_0(s_i) = \| s_i \|_2$
\State $j(0) = \mbox{arg sup}_{i} \hat{\sigma}_0(s_i)$
\State Define $\hat{\sigma}_0 = \hat{\sigma}_0(s_{j(0)}) $
\State $q_1 = s_{j(0)} / \| s_{j(0)} \|_2$
\While{$\hat{\sigma}_k \geq \tau$}
\State $k = k + 1$
\State $ \hat{\sigma}_k( s_i ) = \| s_i - Q_k Q_k^T  s_i \|_2$ $\forall s_i \in S$
\State $ j(k) = \mbox{arg sup}_{i}  \hat{\sigma}_k ( s_i ) $
\State $ \hat{\sigma}_k = \hat{\sigma}_k ( s_{j(k)} ) $
\State $ q_{k+1} = s_{j(k)}  - Q_k Q_k^T  s_{j(k)}$
\State $ q_{k+1} := q_{k+1} / \| q_{k+1} \|_2$
\EndWhile
\end{algorithmic}
\end{algorithm}
}
\end{minipage}

\hspace{0.5cm}

\begin{proposition}\label{prop:MGS_RBgreedy}
For a matrix $S \in \mathbb{R}^{N\times M}$,
the RB-greedy Algo.~\ref{algo:RB-greedy} is equivalent to the MGS with column pivoting Algo.~\ref{algo:POD-QR}. 
Moreover, the diagonal entries of $R$ are non--increasing, i.e., 
\begin{align} \label{eq:Rprop}
    R(1,1) \geq R(2,2) \geq \dots \ge R(M,M) \geq 0 , 
\end{align}  
\end{proposition}
\begin{ieee}
\begin{proof}
It is sufficient to show their pivoting strategies are equivalent. In Algo.~\ref{algo:RB-greedy}, the $k$th column pivot is 
\begin{align*}
j(k) = \mbox{arg max}_{i} \| s_i - \sum_{j=1}^{k-1} q_j^Ts_iq_j \|_2 \, .
\end{align*}
For MGS with column pivoting Algo.~\ref{algo:POD-QR} we have
\begin{align*}
\mbox{arg max}_{i} \| v^{(k)}_i \|_2 \, ,
\end{align*}
where $v^{(k)}_i$ is defined from a recurrence relation
\begin{align*}
v^{(k)}_i \leftarrow \zeta := v^{(k)}_i - q_{k}^T v^{(k)}_i q_{k} \, , \qquad v^{(1)}_i = s_i \, .
\end{align*}
Thus using the orthogonality of $q_k$ we obtain 
\[
  \zeta = v_i^{(1)} - \sum_{j=1}^{k-1} q_j^T v_i^{(1)} q_j 
        = s_i - \sum_{j=1}^{k-1} q_j^T s_i q_j 
\]
Hence, the selection of pivots
\[
\mbox{arg max}_{i} \| s_i - \sum_{j=1}^{k-1} q_j^T s_i q_j \|_2 = \mbox{arg max}_{i} \| v^{(k)}_i \|_2 \, ,
\]
are equivalent. 
\end{proof}
\end{ieee}

\begin{ieee}
\begin{rem}\label{rem:MGS_Greedy_cost}
\rm
As presented, the MGS with pivoting carries a greater memory overhead. In terms of operation counts the MGS steps 6 and 12 are dominant, requiring $2N(M-j+1)$ and $4N(M-j)$ FLOPs at iteration $j$. Then the total count is
\begin{align}
\sum_{j=1}^k (4N)(M-j)+2N(M-j-1) & \approx \sum_{i=1}^k 6NM -\sum_{j=1}^k (6Nj) = 6kNM-3Nk^2 \, ,
\end{align}
after $k$ steps. The RB-greedy's dominant cost is the pivoting step 10, where, after $k$ iterations, the accumulated FLOP count is 
\begin{align}
\sum_{i=1}^k (3N(i-1)+N+2N)M = \sum_{i=1}^k 3NMi \approx \frac{3}{2}k^2NM \, . 
\end{align}
\end{rem}

\begin{rem}\rm
These algorithms may be modified to improve memory overhead, FLOP counts, or conditioning. For very large problems the dominant FLOP count of Algo.~\ref{algo:RB-greedy} 
can be dramatically reduced if one stores the projections $Q_k Q_k^T s_i$ from each previous step; this would essentially amount to storing a matrix $V$ as is done in the MGS with pivoting. Furthermore, the naive implementation of the classical Gram-Schmidt procedure can lead to a numerically 
ill-conditioned algorithm. To overcome this one should use well-conditioned orthogonalization algorithms such as the iterated Gram-Schmidt~\cite{Giraud,Hoffmann_IMGS,Ruhe1983591} or Householder reductions. Our implementation of Alg.~\ref{algo:RB-greedy} uses Hoffmann's iterated Gram-Schmidt~\cite{Hoffmann_IMGS} which maintains orthogonality for extremely large basis sets~\cite{field2012towards}. We are unaware of results which characterize the preservation 
of the subspace spanned by the original vectors, but for many approximation-driven applications this is not strictly necessary.
\end{rem}
\end{ieee}

The proof shows their pivoting strategies are equivalent. 
Having demonstrated the (finite dimensional) equivalence of Algorithms~\ref{algo:POD-QR} and~\ref{algo:RB-greedy}, we now discuss their properties. Recall that diagonal components of $R$ are non-increasing. The following result, which motivated the stopping criterion in Algo.~\ref{algo:POD-QR}, shows how the diagonal entires of $R$ are closely connected with the max-norm approximation error used in Algorithm~\ref{algo:RB-greedy} (see also Theorem~\ref{theorem:qrmax}).
\begin{corollary} \label{cor:qrstop}
The stopping criterion used in Algorithms~\ref{algo:POD-QR} and \ref{algo:RB-greedy} 
fulfills
\begin{align*}
  \max_{1 \leq i \leq M}  \| s_i - Q_k Q_k^T s_i\|_2 = R(k+1,k+1) \, .
\end{align*}
\end{corollary}

\begin{ieee}
\begin{proof}
Theorem~\ref{theorem:qrmax} exactly characterizes the error $\| \widetilde{r}_i \|_2$. Then using the definition of $R(k+1,j)$, $j = 1, \dots, M$ it is not difficult to see that
\[
\| \widetilde{r}_i \|_2 \leq R(k+1, k+1), \ 1 \leq i \leq M .
\]
In addition, we pick $R(k+1,k+1)$ such that
\begin{align*}
    R(k+1,k+1) = \max_{1 \leq i \leq M} \| S(:,i) - \sum_{j=1}^{k} Q(:,j)^TS(:,i) Q(:,j) \|_2 \, ,
\end{align*}
and since the diagonal entries of $R$ are nonincreasing, we obtain
\[
  R(k+1,k+1) \le R(k,k) = \max_{1 \leq i \leq M} \| S(:,i) - \sum_{j=1}^{k-1} Q(:,j)^TS(:,i) Q(:,j) \|_2
\]
Then the estimate follows after using Theorem~\ref{theorem:qrmax}. 
\end{proof}
\end{ieee}

We next derive a representation of the estimate in Corollary~\ref{cor:qrstop} in terms of the singular values of $S$. 
\begin{corollary}
Let $S \in \mathbb{R}^{N \times M}$ then
\[
 \max_{1 \leq i \leq M}  \| s_i - Q_k Q_k^T s_i\|_2 
  = \Big(\prod_{i=1}^{k+1} \sigma_i \Big) / 
    \Big( \prod_{i=1}^{k} R(i,i) \Big) . 
\]
\end{corollary}
\begin{ieee}
\begin{proof}
Denote by $S_{k+1}$ the first $k+1$ columns of $ S \Pi$. Taking QR decomposition of $S_{k+1}$, we obtain
\begin{align*}
    S_{k+1} = QR, \quad Q \in \mathbb{R}^{N \times (k+1)}, \ R \in \mathbb{R}^{(k+1) \times (k+1)} ,  
\end{align*}
yielding $\det(S_{k+1}) = \det(R)$. Since $\det(S_{k+1}) = \prod_{i=1}^{k+1} \sigma_i$ and $\det(R) = \prod_{i=1}^{k+1} R(i,i)$, using Corollary~\ref{cor:qrstop} we arrive at the assertion.  
\end{proof}
\end{ieee}

\begin{ieee}
\subsubsection{A reconstruction approach to QR} \label{s:recon}

We recall that the POD algorithm (\ref{algo:POD}) generates the optimal 2-norm basis (see Theorem~\ref{thm:POD}). The goal of this section is to augment the QR with a reconstruction technique such that resulting approximation of $S$ is as accurate as POD, however the algorithm is cheaper than performing an SVD. We note that our method shares some similarities with those algorithms that perform QR decompositions as a precursor to finding the SVD (e.g.~\cite{constantine2014model,trefethen1997numerical}).

\begin{theorem} \label{thm:reconstruction}
Given $S \in \mathbb{R}^{N \times M}$ with ordinary rank $k$, if $X = Q \overline{V}$, then $X_k = Q_k \mathcal{V}$, and 
\[ 
    \| S-X_j X_j^TS \|_2 = \sigma_{j+1} = \| S - Q_j Q_j^T S \|_2
    = \| S - V_j V_j^TS \|_2 , \quad 1 \leq j \leq k.
\]  
\end{theorem}

\begin{proof}
Since $S$ has ordinary rank $k$, using Corollary~\ref{cor:oqrkrank}
\begin{align*}
R =
\left[ \begin{array}{cc}
R_{11} & R_{12} \\
0 & 0  \end{array} \right]
= \overline{V} \Sigma \overline{W}^T . 
\end{align*}
Let $R = \overline{V} \Sigma \overline{W}^T$ be the SVD of $R$ and 
\begin{equation*} 
\mathcal{V} \Sigma_{k,M} \mathcal{W}^T = 
\left[ \begin{array}{cc}
R_{11} & R_{12} \end{array} \right] ,
\end{equation*}
be \emph{full SVD} of $\left[ \begin{array}{cc} R_{11} & R_{12} \end{array} \right]$, where $\mathcal{V} \in \mathbb{R}^{k \times k}, \ \Sigma_{k,M} \in \mathbb{R}^{k \times M}$, and $\mathcal{W} \in \mathbb{R}^{M \times M}$. 
Then
\begin{align*}
R =
\left[ \begin{array}{cc}
R_{11} & R_{12} \\
0 & 0  \end{array} \right]
= \overline{V} \Sigma \overline{W}^T
=
\left[ \begin{array}{cc}
\mathcal{V} & 0 \\
0 & \widetilde{\mathcal{V}} \end{array} \right] 
\left[ \begin{array}{cc}
\Sigma_{k,M}  \\
0  \end{array} \right] \mathcal{W}^T,
\end{align*}
where $\widetilde{\mathcal{V}} \in \mathbb{R}^{(N-k) \times (M-k)}$ is an arbitrary orthogonal matrix.
Defining
\begin{align*}
X = Q \overline{V} = 
\left[ \begin{array}{cc}
Q_k & Q_{N-k} \end{array} \right] 
\left[ \begin{array}{cc}
\mathcal{V} & 0 \\
0 & \widetilde{\mathcal{V}} \end{array} \right]  =
\left[ \begin{array}{cc}
Q_k \mathcal{V} & Q_{N-k} \widetilde{\mathcal{V}} \end{array} \right]
\end{align*}
leads to
\begin{align*} 
X_j = Q_k \mathcal{V}_j, \quad \mbox{for }         1 \le j \le k . 
\end{align*}
We now show that $ \| S - X_j X_j^T S \|_2 = \sigma_{j+1}$ for $1 \leq j \leq k$, the other estimate is due to Theorem~\ref{thm:RRQR}. Recalling $S = QR$ and $R = \overline{V}\Sigma\overline{W}^T$, we deduce 
\begin{align*}
    \| S - X_j X_j^T S \|_2 
    &= \| X \Sigma \overline{W}^T- X_j X_j^T X \Sigma \overline{W}^T \|_2 
     = \| X \Sigma - X_j^0 \Sigma \|_2 \\
    &= \| X \Sigma - X \Sigma_j \|_2 = \| \Sigma - \Sigma_j^0 \|_2 
    = \sigma_{j+1} 
    = \| S - V_j V_j^TS \|_2,
\end{align*}
where $ X_j^0 $ is $ X $ with zeros in columns $ j+1 $ to $ N $ and the last equality is due to Theorem~\ref{thm:POD}(ii). 
\end{proof}

With this motivation, we introduce the following algorithm for a matrix $S \in \mathbb{R}^{N\times M}$ of \emph{numerical rank} $k$. 
\hspace{0.5cm}

{\scriptsize
\begin{algorithm}[H]
\caption{Reconstruction Algorithm}
\label{algo:reconstruction}
\begin{algorithmic}[1]
\State \textbf{Input}: Samples $S = (s_1, ... , s_M) \in \mathbb{R}^{N \times M}$, tolerance $ \tau_1, \tau_2 > 0 $
\State \textbf{Output}: $ X_k = (x_1, ... , x_k) \in \mathbb{R}^{N \times k}$
\vskip 10pt
\State Perform a partial $j$-term MGS with pivoting of $S$ stopping whenever $|R(j,j)| < \tau_1$
\State Let $S = Q_j R(1:j,1:M)$ be the result of step 3
\State Perform an SVD $\mathcal{V} \Sigma \mathcal{W}^T = R(1:j,1:M) $
\State Find $k \in \{1, \dots, j\}$ such that $\sigma_{k+1}<\tau_2<\sigma_k$ 
\State Return $ X_k = Q_k \mathcal{V}(1:k,1:k)$ 
\end{algorithmic}
\end{algorithm}
}

\begin{rem}\rm
Computational cost of step 3 is $\mathcal{O}(Mj^2+Nj^2)$, which is less than the MGS with pivoting cost $\mathcal{O}(jNM)$ given that $j \ll N$. Steps 4-5 enrich the basis: generally speaking, the diagonal entry of $R$ decrease slower than the singular values, and so we employ the SVD basis for improved accuracy.
\end{rem}

Next, we give and error estimate for the Reconstruction algorithm. To this end we write $S=S_1+S_2$, such that
\begin{align*}
S_1 = \left[ \begin{array}{cc} Q_k & Q_{N-k} \end{array} \right] 
\left[ \begin{array}{cc}
R_{11} & R_{12} \\
0 & 0 \end{array}  \right], \
S_2 = \left[ \begin{array}{cc} Q_k & Q_{N-k} \end{array} \right] 
\left[ \begin{array}{cc}
0 & 0 \\
0 & R_{22} \end{array}  \right] , 
\end{align*}
where we have assumed that the \emph{full QR} of $S = QR$.
In addition, let 
\begin{align*}
S_1 = V_1 \Sigma_1 W_1^T, \ S_2= V_2 \Sigma_2 W_2^T
\end{align*}
be the \emph{full SVD} of $S_1, S_2$ with singular values $\sigma(S_1)_{j}, \sigma(S_2)_{\ell}$. 
\begin{lemma} \label{lem:recon_jj}
Let $S \in \mathbb{R}^{N \times M}$ has numerical rank $k$ and if at $k^{th}$ step in Algo.~\ref{algo:reconstruction} $X_k = Q_k \mathcal{V}_j$, then 
\begin{align*}
    \| S - V_j V_j^T S \|_{2} 
    & \le \| S-X_jX_j^T S \|_2  \nonumber \\
    & \le \| S_1 - X_j X_j^T S_1 \|_2 +\|R_{22} \|_2, \quad 
            \mbox{for } 1 \leq j \le k.
\end{align*}
\end{lemma}
\begin{proof}
The first inequality follows immediately due to the optimality of the POD. To derive the second inequality we use $S = S_1 + S_2$ and readily obtain 
\begin{align*}
\| S - X_j X_j^T S \|_2 
    &= \| S_1- X_j X_j^T S_1 + S_2 - Q_k \mathcal{V}_j \mathcal{V}_j^T Q_k^T S_2 \|_2 \\  \nonumber
    & \leq \| S_1 - X_j X_j^T S_1 \|_2 + \| (I - Q_k \mathcal{V}_j \mathcal{V}_j^T Q_k^T) S_2 \|_2 .
\end{align*}
Using the definition of $S_2$ we obtain
\begin{align*}
\| S - X_j X_j^T S \|_2    
    &\le \| S_1 - X_j X_j^T S_1 \|_2 + \|I-Q_k \mathcal{V}_j \mathcal{V}_j^T Q_k^T \|_2 \| R_{22} \|_2. 
\end{align*}
As $Q_k \mathcal{V}_j \mathcal{V}_j^T Q_k^T$ is a projector, hence (see \cite{JXu_LZikatanov_2003a,chaturantabut:2737,DBSzyld_2006a}) we obtain 
\[
    \| I - Q_k \mathcal{V}_j \mathcal{V}_j^T Q_k^T \|_2 = \| Q_k \mathcal{V}_j \mathcal{V}_j^T Q_k^T \|_2 , 
\]
thus we conclude.
\end{proof}

\begin{theorem}\label{thm:recon_jj}
Under the assumption of Lemma~\ref{lem:recon_jj}, for $1 \leq j \le k$, 
the following estimate holds
\begin{align}\label{eq:recon_up}
\| S - V_j V_j^T S \|_{2} \le
    \| S-X_jX_j^T S \|_2 
    &\leq \| S_1 - Q_j Q_j^T S_1 \|_2 +\|R_{22} \|_2 \nonumber \\
    &= (\sigma_1)_{j+1} +\|R_{22} \|_2  .
\end{align}
\end{theorem}
\begin{proof}
Recall that 
\begin{align*}
S_1 = \left[ \begin{array}{cc} Q_k & Q_{N-k} \end{array} \right] 
\left[ \begin{array}{cc}
R_{11} & R_{12} \\
0 & 0 \end{array}  \right]
\end{align*}
In view of Lemma~\ref{lem:recon_jj}, it is sufficient to realize the result of Theorem~\ref{thm:reconstruction} for $S_1$. 
The proof follows exactly the same way as in Theorem~\ref{thm:reconstruction} and is omitted for brevity. 
\end{proof}

\begin{rem}[full QR and reconstruction]
\rm
For the full QR, we recall that
\[
    \| S - Q_kQ_k^T S \|_2 = \| QR - Q_k^0 R \|_2 = \|R_{22} \|_2 . 
\]
which is sharper than the reconstruction estimator in \eqref{eq:recon_up}. However, we emphasize that the bound in \eqref{eq:recon_up} is an estimate and, in addition, the reconstruction algorithm is tractable for large matrix sizes. 
\end{rem}

\begin{rem}[POD and reconstruction]
\rm
We note that the error bound in \eqref{eq:recon_up} has two contributions. If $\| R_{22} \|_2$ is 
of order $\epsilon_{\tt mach}$
then the reconstruction \eqref{eq:recon_up} and POD (Theorem~\ref{thm:POD}) error behave similarly.
\end{rem}
\end{ieee}

\section{Large-scale QR+DEIM code for model reduction} \label{s:numerics}

\subsection{Overview}

We now describe an implementation of the greedy algorithm~\ref{algo:RB-greedy}
for finding a column-pivoted $QR$ decomposition of a complex-valued, dense matrix.
Our publicly available code
greedycpp~\cite{greedycpp} 
has been developed over the past $3$ years and has been applied to
a variety of production-scale problems. For example, it has been used
to build reduced order quadrature rules~\cite{HAntil_SField_RHNochetto_MTiglio_2013}, 
which can be used to 
accelerate Bayesian inference studies~\cite{canizares2015accelerated,smith2016fast}
and provide low-latency parameter estimation (see 
the supplemental material of Ref.~\cite{PhysRevLett.118.221101}).

As is typical in model reduction applications, given a model $M$, 
we shall interpret the column $s_i = M({\bf x};\nu_i)$ 
as stemming from the model's evaluation at a parameter value $\nu_i$ and the rows
as the evaluations on a grid discretizing the relevant independent variable ${\bf x}$. 
Alg.~\ref{algo:RB-greedy} identifies the 
specially selected columns (the pivots) whose span is the reduced model (i.e. basis) space. 
Our code allows the user to set an approximation error threshold ($\tau$ in
algorithms \ref{algo:RB-greedy} and \ref{algo:POD-QR}) such that, according to Corollary 5.6, 
guarantees all columns $s_i$ satisfy $\| s_i - Q_k Q_k^T s_i\|_2 < \tau$.

In addition to the dimensional reduction
feature, the code also selects a set of empirical interpolation (EI) nodes using a
fast algorithm (see Alg.~5 of Ref.~\cite{HAntil_SField_RHNochetto_MTiglio_2013})
and performs out-of-sample validation of the resulting basis and empirical interpolant.
The basis, pivots, and EI nodes can be exported to different file formats including
ordinary text, (GSL) binary, and NumPy binary. The code's simple interface
allows any model written in the C/C++ 
language to be used. Supporting scripts and input files allow for
control over the run context. The code can also perform 
basis validation on quadrature grids that differ from the one used to form $S$
and automatically enrich the basis by iterative refinement~\cite{smith2016fast}.

While there has long been an interest in designing efficient parallel
algorithms of column pivoting QR factorization, 
as far as production-scale publicly available codes, however, 
we are only aware of the (Sca)LAPACK routines.
Furthermore, because (Sca)LAPACK is a general purpose 
linear algebra library, it is insufficient for many model-reduction-type 
tasks, which is a key reason why we pursued our own implementation. 
Finally, we note that because column pivoting limits potential parallelism in
the QR factorization~\cite{gallopoulos2015parallelism}, algorithmic
design remains an area of active research~\cite{demmel2015communication}.
Earlier works proposed BLAS-3 versions of a parallelized algorithm~\cite{Quintana-Ortí1997}
and communication-avoiding local pivoting strategies~\cite{bischof1991parallel}.

\subsubsection{Gravitational wave model}

We briefly describe the particular model used to form our 
snapshot matrix $S$.
We use the IMRPhenomPv2 model~\cite{PhysRevLett.113.151101}~\cite{PhysRevLett.113.151101,Husa:2015iqa,Khan:2015jqa} 
of gravitational waves emitted by two merging binary black holes. 
The model is implemented 
as part of the publicly available LIGO Analysis Library (available, e.g., at \texttt{https://github.com/lscsoft/lalsuite}).
Our code fills the matrix $S$ by calls to the IMRPhenomPv2 model
without any file I/O.
The parameter values that define $S$ 
are distributed among the different MPI processes, and each 
process is responsible for 
forming a ``slice" of $S$ over a subset of columns. This strategy 
is useful for computationally-intensive models. 

\subsubsection{Serial code}

Pivoted QR naturally decomposes into two parts, a {\em pivot search} followed by 
{\em orthogonalization}. Consider Alg.~\ref{algo:RB-greedy}
at iteration $k$. The pivot search proceeds by computing 
all $M$ local error residuals $\hat{\sigma}_k(s_i)$ and finding
$J = \arg \sup_i \hat{\sigma}_k(s_i)$. The $J^{th}$ column, $s_J$,
is the one to be pivoted, and an orthogonalized $s_J$
becomes the next column of $Q$. 

Due to the orthogonality of the basis vectors, the relationship 
\begin{align} \label{eq:rec1}
Q_k Q_k^T  s_i = Q_{k-1} Q_{k-1}^T  s_i + q_{k} q_{k}^T  s_i \, ,
\end{align}
can be exploited to yield constant complexity at each iteration provided we retain information from
the previous iteration, namely $ Q_{k-1} Q_{k-1}^T  s_i$. We neither store this vector 
(which would increase the algorithm's memory footprint) nor 
replace the columns of $S$ (which requires additional operations of order $N$).
Notice that
\begin{align}
\| s_i - Q_k Q_k^T  s_i \|_2^2 = \|s_i \|_2^2 - \sum_{j=0}^k c_j c_j^* \, ,
\end{align}
where $c_j = q_j^T s_i$ is the inner product between $s_i$ and the $j^{th}$ basis.  
The relationship \eqref{eq:rec1} becomes 
\begin{align} \label{eq:Residual}
\hat{\sigma}^2_k(s_i) = \|s_i \|_2^2 - \left[ \sum_{j=0}^{k-1} c_j c_j^* + c_k c_k^* \right] \, ,
\end{align}
and, furthermore, the square-root need not be taken in order to identify the next pivot. To avoid
catastrophic cancellation error, we store values for
both $\|s_i \|$ and $\sum_{j=0}^{k-1} c_j c_j^*$. As the sum consists of non-negative terms,
the update $\sum_{j=0}^{k-1} c_j c_j^* \rightarrow \sum_{j=0}^{k} c_j c_j^*$ is well-conditioned.
The $j^{th}$ pivot search has a constant complexity with an asymptotic
FLOP count of $\bigo{2MN}$. Fig.~\ref{fig:SinglePivot} (left) plots the scaled time of
$T_{j}^{\rm pivot}/N$ with the iteration index $j$ for different values of $N$.

Next, the newly selected pivot column is orthogonalized with respect
to the existing basis set to yield the next basis $q_k$, which, in turn, will be used 
to find the $c_k$ coefficients required to update Eq.~\eqref{eq:Residual}. 
Both the classical and modified
Gram-Schmidt algorithms suffer from poor conditioning~\cite{Hoffmann_IMGS,bjorck1967solving}.
We use the iterative modified Gram-Schmidt (IMGS) algorithm 
of Hoffman~\cite{Hoffmann_IMGS} (Ref.~\cite{Hoffmann_IMGS}'s ``MGSCI" algorithm with $\kappa = 2$)
for which Hoffmann conjectures an orthogonality relation $\| I - Q_k^T Q_k \|_2 \approx \kappa \epsilon_{\tt mach} \sqrt{M}$. 
Orthogonalization of the $j^{th}$ basis using Hoffman's IMGS has an asymptotic FLOP count of
$\bigo{\nu_j jN}$, where $\nu_j$ is is the number of MGS iterations, which depends on $j$, and 
is typically less than $3$.
Fig.~\ref{fig:SingleOrtho} (right) plots the scaled time 
$T_{j}^{\rm IMGS}/N$ with iteration index $j$ for different values of $N$.

To summarize, our implementation of Alg.~\ref{algo:RB-greedy} requires
$\bigo{2MNk+ \frac{1}{2} \hat{\nu} N k (k+1)}$ operations
to find $k$ basis, where $\hat{\nu}$ is an ``effective" value of $\nu_j$.

Unless noted otherwise, our timing experiments
have been carried out on either the San Diego Supercomputer 
Center's machine Comet (each compute node features two 
Intel Xeon E5-2680v3 2.5 GHz chips, each equipped with 12 cores,
and connected by an InfiniBand interconnect)
or the National Center for Supercomputing Applications' machine 
Blue Waters (each compute node features 32 OS-cores, and every 2 of these 
OS-cores share a single floating point unit. Nodes are connected by 
the Gemini interconnect).
We have made only modest attempts at core-level optimization 
which, importantly, includes the use of Advanced Vector Extensions (AVX2)
for vector-vector products. 

\begin{figure}[ht]
\centering
\subfigure[Single core pivot search.]{
  \includegraphics[width=0.47\linewidth]{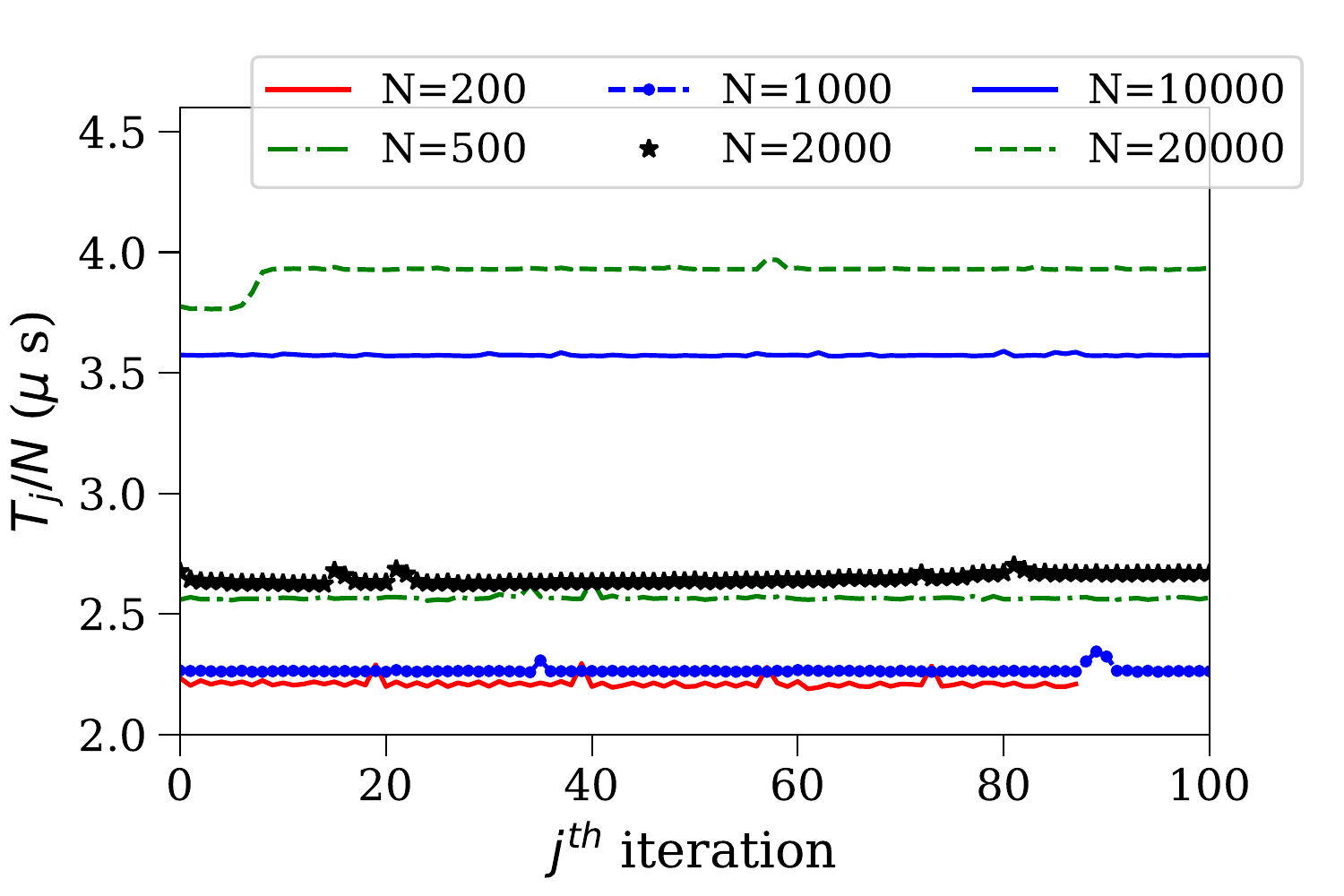}
  \label{fig:SinglePivot}}
\subfigure[Single core orthogonalization.]{
  \includegraphics[width=0.47\linewidth]{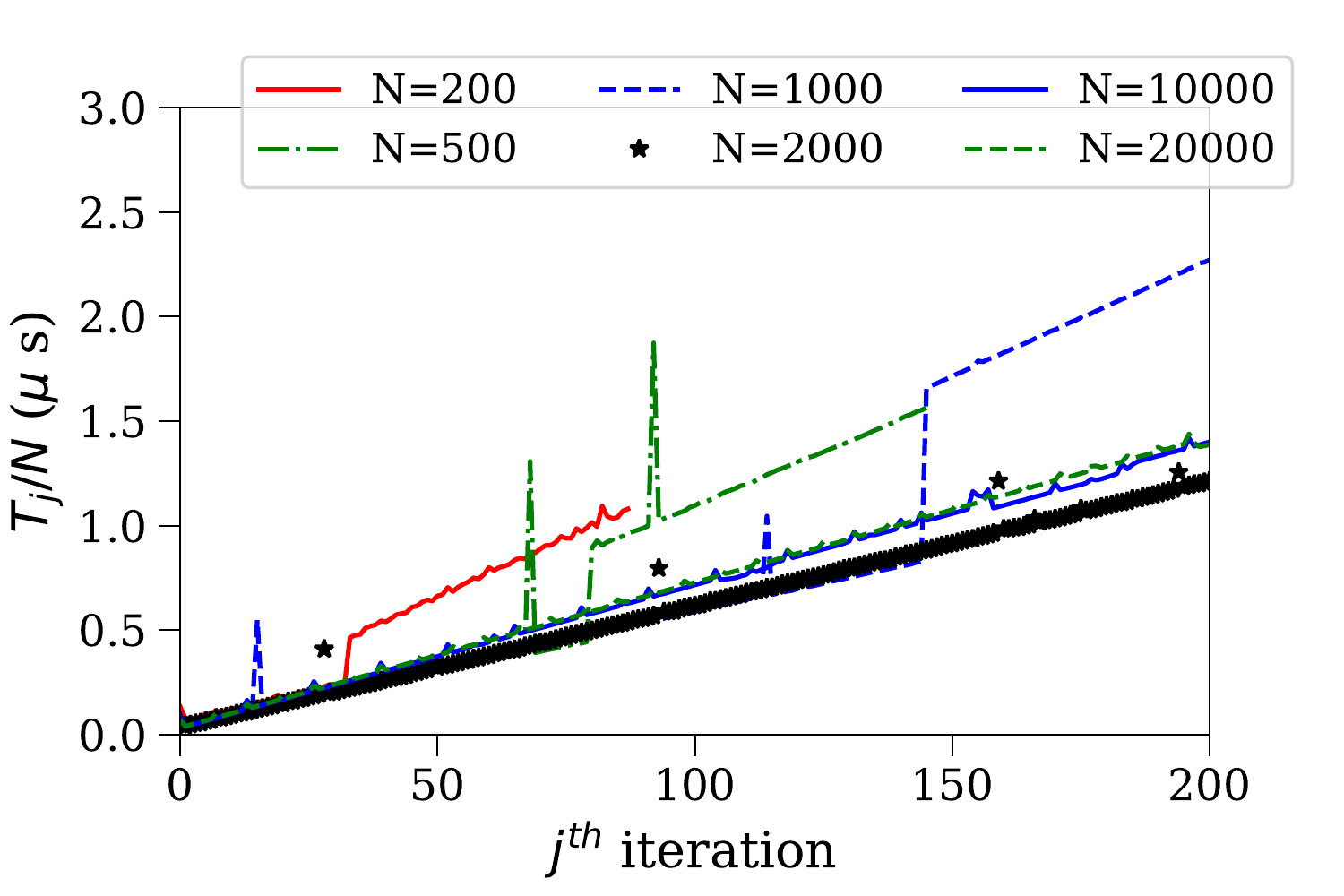}
  \label{fig:SingleOrtho}}
\caption{{\bf Left:} Pivot search time as a function of iteration index $j$.
Based solely on operation counts one should expect $T_{j}^{\rm pivot}/N$ to be 
independent of $N$. Differences are likely due to data access latencies which do depend on $N$.
{\bf Right:} Orthogonalization time as a function of iteration index $j$.
As expected, these quantities exhibit a linear growth with $j$. Timing data from 
the two smallest values of $N$ ``jump" due to additional orthogonalization iterations ($\nu$ increases from $1$ to $2$).}
\label{fig:SingleCore}
\end{figure}

\subsubsection{Parallel code}
\label{sec:NumericalSharedMem}

We consider a natural parallelization-by-column strategy. 
Each core~\footnote{We consider parallelization
with MPI (``by process"), OpenMP (``by thread") and an MPI/OpenMP hybrid. To streamline the presentation, 
we avoid the terms ``process" and ``thread" in favor of ``core". Since we will never run more than one
process or thread per core, the terminology should be unambiguous and clear from context.
The book {\em Introduction to high performance computing for scientists and engineers}
provides a comprehensive introduction to many of the high performance computing concepts
discussed through this section~\cite{hager2010introduction}.}
is given a subset of columns to manage, and a separate ``master" core
is responsible for all orthogonalization activities. 
We denote $P_{\rm pivot}$ as the number of cores devoted 
to the pivot search and $P_{\rm ortho}$ as the number of cores devoted to basis orthogonalization. 
Load balancing is trivially accomplished by distributing $N/P_{\rm pivot}$ columns of $S$ among
$P_{\rm pivot}$ cores. Each pivot core loads or creates its chunk of $S$ in parallel.
Parallelization of the orthogonalization portion of the algorithm will be discussed later; 
for now $P_{\rm ortho}=1$.

The $j^{th}$ iteration is initiated after the orthgonalization core broadcasts the $j-1$ basis 
vector to all $P_{\rm pivot}$ cores. Next, each pivot core computes its contribution of 
Eq.~\eqref{eq:Residual} and its maximum.
This information is communicated to all pivot and orthgonalization cores. The pivot core
with the global maximum residual error sends its column to the orthogonalization
core to orthogonalize.

We model the $j^{th}$ iteration's computational cost as
\begin{align} \label{eq:CostModel}
T_j = T_{j}^{\rm pivot} + T_{j}^{\rm IMGS} + C_j \, ,
\end{align}
where $T_{j}^{\rm IMGS}$  measures the orthogonalization time, 
$T_j$ is the entire while loop appearing in Alg.~\ref{algo:RB-greedy},
and $C_j$ includes additional parallelization overheads such as any communication cost 
and/or thread-management overheads. 
As equality holds to at worst $1\%$ (typically $0.001\%$),
we often report only $T_{j}^{\rm pivot} + C_j$ and $T_j$. 
All timing measurements are made from the master process,
and the timer measuring $T_{j}^{\rm pivot}$ starts before the next basis vector
is broadcasted to all the workers and ends after the next $j+1$ (unorthogonalized) 
column basis has been received by the master process.
Similar to the single core case (cf.~Fig.~\ref{fig:SinglePivot}) we 
observe (as expected) $T_{j}^{\rm pivot} + C_j$ to be independent
of $j$ and, therefore, often report values at some fixed value of $j$. 

We consider parallelization by message passing interface (MPI) and OpenMP.

{\bf MPI.} Each MPI process runs on a unique core. The pivot and orthogononalization cores 
communicate global pivot information using MPI\_Allreduce(),
the selected column vector is passed to the orthgonalization work using MPI\_Send()
and MPI\_BCast() provides all the pivot cores 
with this new orthonormal basis~\footnote{We experimented with a few different MPI library
functions, such as broadcast, reduction and gather, but found these to perform worst. The 
code's git history documents these experiments, which are not reported here.}.

{\bf OpenMP.} OpenMP uses threads to parallelize a portion of the code using 
the fork-join model. When using OpenMP, we define a large parallel region construct
using $P_{\rm pivot}$ threads and 
enclosing the entire while loop; in fact most of the worker's code is inside of the 
parallel region. 
We found this to given better performance results 
as compared with parallelizing the for-loop over columns, possibly because a wider 
parallelized region avoids multiple fork-joins.
A designated master thread carries out the orthogonalization task.

Figure~\ref{fig:SingleNode} considers a family of strong scaling tests where the matrix
size is fixed and we vary the number of cores from $1$ to $24$. 
The left panel reports the parallelization efficiency
\begin{align}
E_C = \frac{T_1}{C T_C} \, ,
\end{align}
for the pivot search parallelized with OpenMP, 
where $C$ is the number of cores, and $T_1$ and $T_C$ denote the walltime 
using $1$ and $C$ cores, respectively. Perfect scalability is achieved whenever $E_C = 1$.
The code's speedup, another often
quoted scalability measure, is simply $S_C = C \times E_C$. 

Consistently high efficiencies 
are observed over a range of problem sizes, with
speedups $\approx 20$ routinely observed. For smaller problem sizes, the efficiency
is reduced as the parallel overhead $C_j$ becomes 
a sizable fraction of the overall cost~\footnote{Very large values of $M$,
say $M \ge 10^{6}$, also shows reduced scalability presumably due to memory access times.}. 
Figure~\ref{fig:SingleNodePivot} shows the pivot search portion of the algorithm 
is efficiently parallelized. Figure~\ref{fig:SingleNodeFull} shows the full algorithm's efficiency. 
Evidently scalability is poor for shorter matrices (small values of $N$), which should be expected from
Eq.~\eqref{eq:CostModel} and Amdahl's law. Approximating the computational
cost to be proportional to the asymptotic FLOP count and assuming 
$M \gg k$, $M > P_{\rm pivot} k$ and $C_j=0$, the efficiency of our algorithm is
\begin{align} \label{eq:PredictedEff}
E \approx 1 - \nu k (P_{\rm pivot}-1) / (2M) \,.
\end{align}

Thus, for good scalability, our problem should require large values of $M$.
Most model reduction applications easily meet this requirement. 
Indeed, model reduction seeks to approximate the 
underlying continuum problem (often with high parametric dimensionality) 
for which $M \rightarrow \infty$,  
while for parametrically smooth models $\sigma_k \propto \exp(-k)$.
Together, these features suggest $M \gg k P_{\rm pivot}$ is 
often satisfied in practice.

\begin{figure}[ht]
\centering
\subfigure[Pivot search portion.]{
  \includegraphics[width=0.46\linewidth]{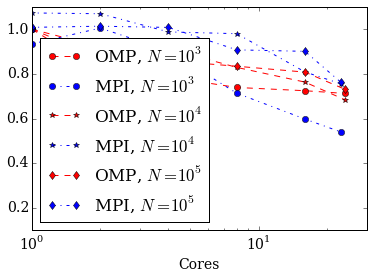}
  \label{fig:SingleNodePivot}}
\subfigure[Entire algorithm.]{
  \includegraphics[width=0.46\linewidth]{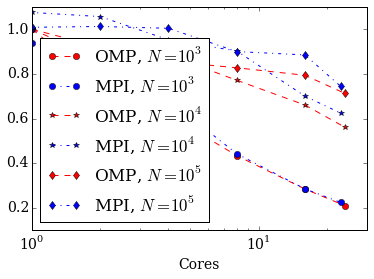}
  \label{fig:SingleNodeFull}}
\caption{Strong scaling efficiency versus cores for a sequence of increasingly ``tall" matrices with
$M=1,000$ and $k=100$ fixed. Our tests were performed on a $24$-core shared memory node of the supercomputer Comet. 
Blocks of columns are distributed to each core with increasingly fewer blocks per core. Specifically 
$1000$, $500$, $250$, $125$, $62$, and $43$ columns are distributed among, respectively, 
$1$, $2$, $4$, $8$, $16$ and $24$ cores. We measure scalability for
the pivot search portion (from the timing data $T_{\rm pivot}$) of the algorithm (left)
and the full algorithm (right).}
\label{fig:SingleNode}
\end{figure}

\subsubsection{Large-core scaling}

Our approach to distributed memory parallelization closely follows that of shared memory parallelization. 
As OpenMP does not support distributed memory environments we cannot use this library for inter-node communication.
We consider two cases. First, a pure-MPI parallelization exactly as described in Sec.~\ref{sec:NumericalSharedMem}.
Second, a hybrid MPI/OpenMP implementation launching one MPI process per 
socket~\footnote{To improve memory access performance, all MPI processes and 
their threads are bound to a socket.}. In turn, each MPI process spawns a team
OpenMP threads. Each thread is responsible for a matrix chunk over which a local pivot search is performed. 
As before we avoid multiple thread fork-joins by enclosing the entirety of the while-loop within
an OpenMP parallel region, with the master thread responsible for all MPI calls (a so-called ``funneled" 
hybrid approach). As only processes participate in inter-node communication, a hybrid code 
potentially reduces the communication overhead as compared to a pure-MPI implementation.
These benefits could become increasingly important at extremely large core counts. 

Figure~\ref{fig:MuliNodeComet} reports on a few scalability tests we ran on
Comet. First, we consider how the pivot search portion of the algorithm scales to large
core counts for a fixed matrix size. Figure~\ref{fig:MultiNodeSSComet}
shows a typical case. We see that going from one core to one node maintains high efficiencies,
which should be expect in light of Fig.~\ref{fig:SingleNodePivot}.
Running on an increasing number of cores means each cores has less work to do (fewer columns 
per core) while the entire algorithm has more communication. As expected,
the efficiency decreases but maintains high values up to $370$ cores. Such observations are matrix-dependent, 
and larger (smaller) matrices are expected to exhibit better (worst) scalability. Interestingly, the MPI/OpenMPI hybrid 
strategy performs much better than pure-MPI for this problem, indicating that the communication overhead can be
somewhat ameliorated; this is also a matrix-dependent observation. Next, we consider how the entire 
algorithm scales to large core counts when the matrix size is also increased commensurate to the number of core; 
this constitutes a weak scaling test. Figure~\ref{fig:MultiNodeWSComet} shows a typical case. We see that the full program's runtime has negligible increase when going from one core to $1,728$ cores (the maximum allowable size on Comet). This demonstrates that very large matrices can be efficiently handled. 

As a final demonstration of our code's capabilities, we repeat the weak scaling test on Blue Waters
where more cores can be used. Figure.~\ref{fig:BlueWaters} shows the same
excellent weak scalability all the way up to $32,768$ cores. In particular, we perform
a column pivoted QR decomposition (to discover the first $k=100$ basis) on a $10,000$-by-$3,276,800$
sized matrix in $13.5$ seconds. For comparison, it took a similar time 
of $11.5$ seconds to QR decompose a much smaller $10,000$-by-$3,200$ matrix using $32$ cores. 

\begin{figure}[ht]
\centering
\subfigure[Strong scaling.]{
  \includegraphics[width=0.46\linewidth]{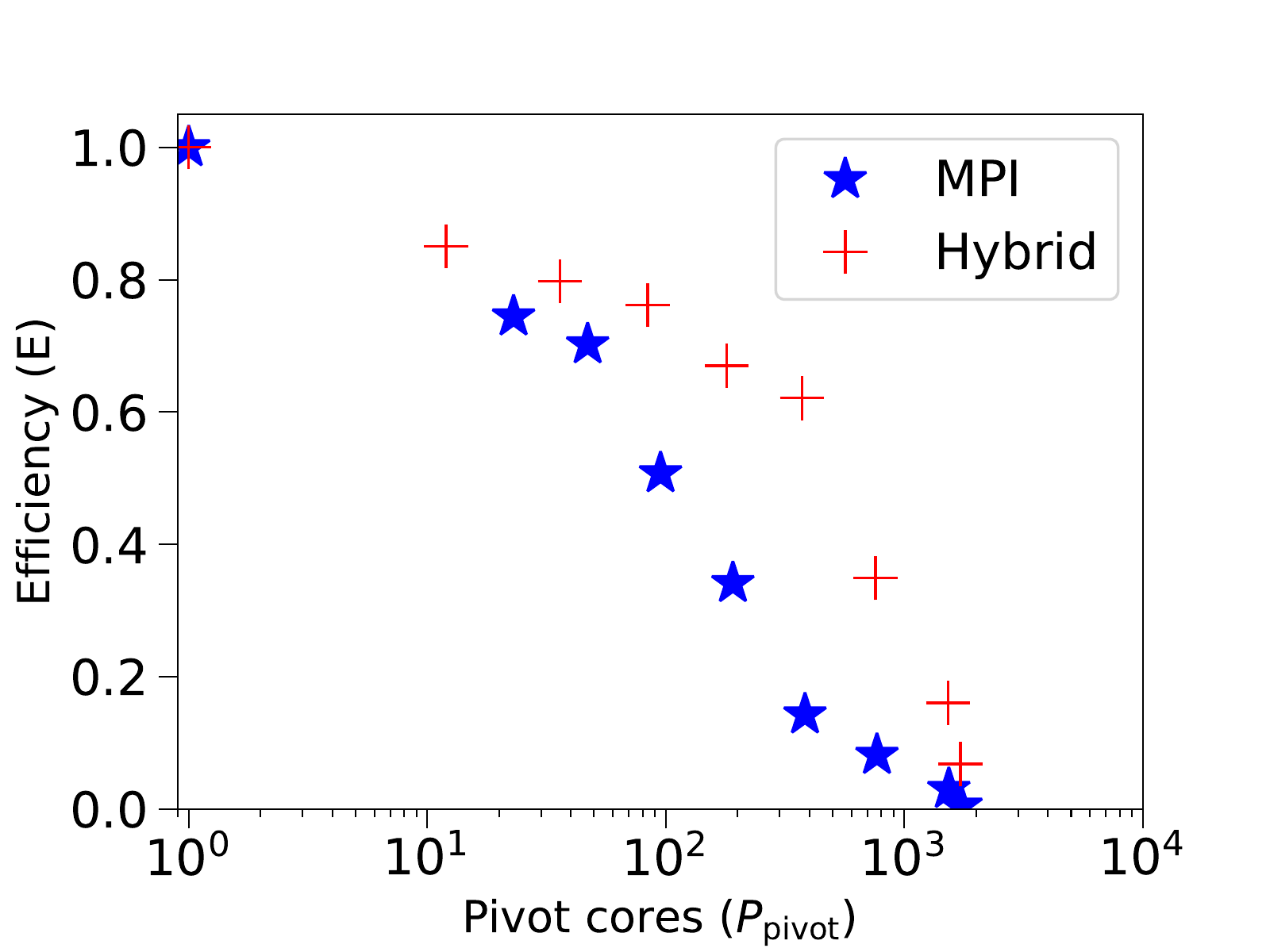}
  \label{fig:MultiNodeSSComet}}
\subfigure[Weak scaling.]{
  \includegraphics[width=0.46\linewidth]{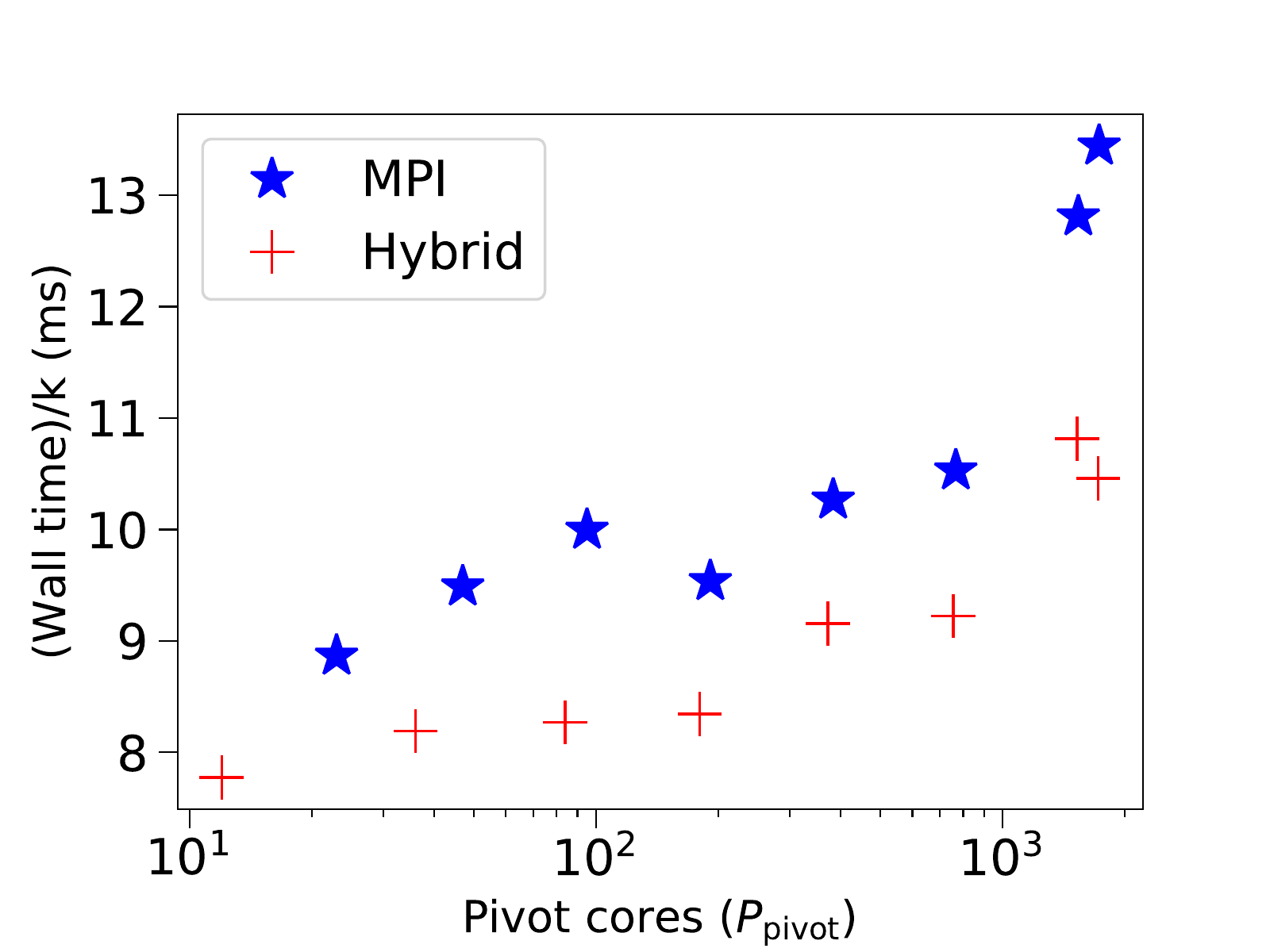}
  \label{fig:MultiNodeWSComet}}
\caption{Scalability on the supercomputer Comet up to $1,728$ cores with $k = 100$.
{\bf Left:} Strong scaling efficiency of the pivot search part of the algorithm 
for a matrix with $N=1,000$ rows and $M=240,000$ columns.  
At  $1728$ cores, each pivot core has $\approx 138$ columns and 
the MPI (Hybrid) efficiency is $.005$ ($.07$), corresponding to 
very small speedups. This is expected; 
at this scale the communication time is a sizable fraction of the overall 
cost, and only increasing the workload (more rows or columns) would increase the efficiency. 
For this particular matrix size, using $370$ cores the hybrid 
parallelization strategy has good efficiencies of about $.62$. Of particular importance
is that the hybrid code is significantly more efficient than the 
pure-MPI code. {\bf Right:} Weak scaling efficiency of the full algorithm for a 
matrix with $N=10,000$ rows
and $M = P_{\rm pivot} \times 100$ columns. The total time is scaled by $k=100$.
The slow growth in the total time with increasingly more cores/columns
demonstrates good weak scalability, allowing very large matrices to be tackled.
(perfect weak scalability would result in a horizontal line.) }
\label{fig:MuliNodeComet}
\end{figure}

\begin{figure}[ht]
\centering
  \includegraphics[width=0.46\linewidth]{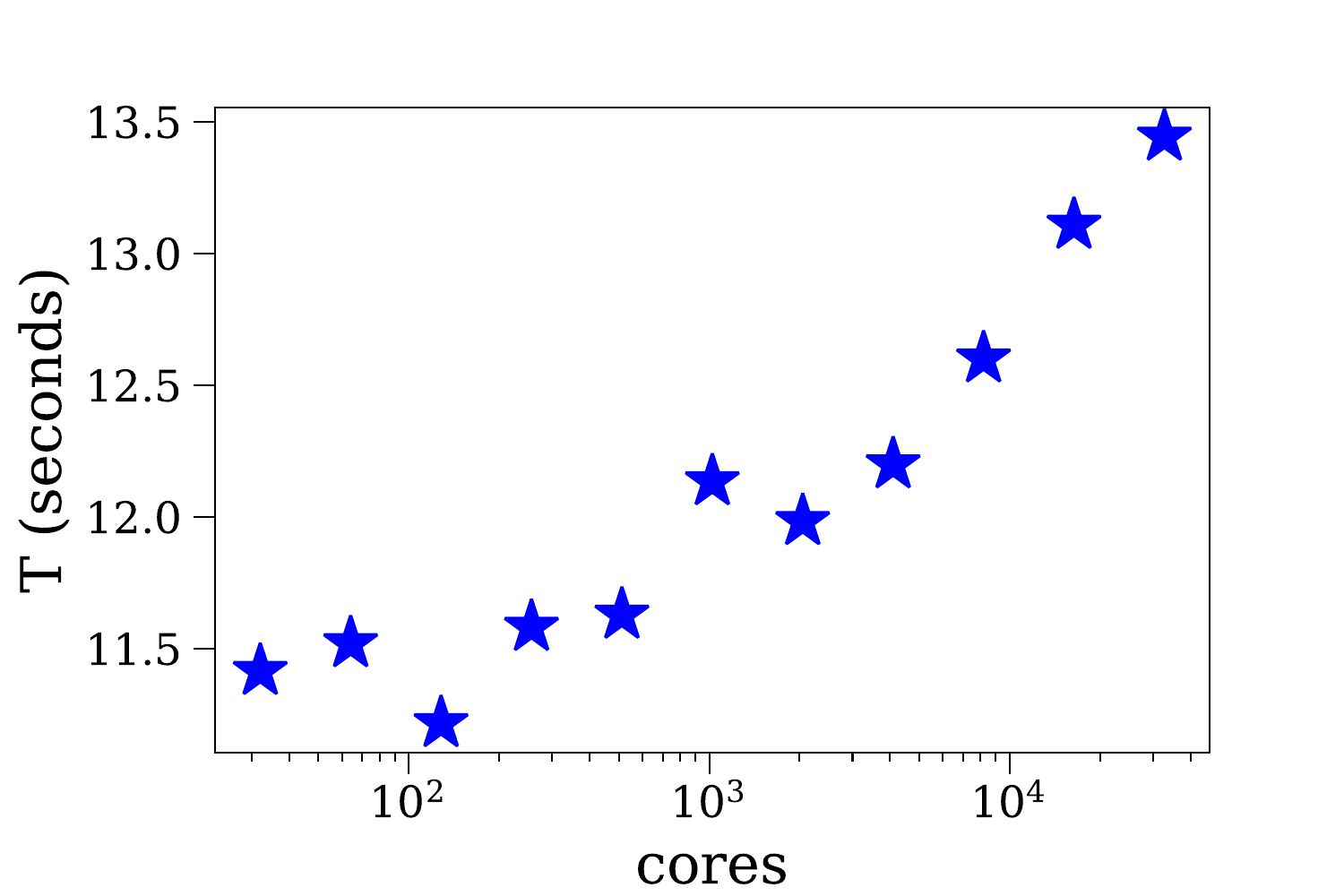}
\caption{Weak scaling on Blue Waters with $N=10,000$ rows and $k=100$. We let
the number of columns $M = 100 * {\rm cores}$ scale linearly with the number of cores.
The number of cores is increased from $32$ to $32,768$, with the largest 
matrix having $3,276,800$ columns. A mild logarithmic growth in the 
full algorithm's time is evident, and the good weak scalability
demonstrates the ability to compute column pivoted QR
decompositions of extremely large matrices.}
\label{fig:BlueWaters}
\end{figure}

\subsubsection{Discussion and Limitations}

We have demonstrated our code is capable of handling very large matrix sizes. 
For smaller sized matrices, 
the orthogonalization routine becomes a performance obstacle. Because the basis is revealed sequentially,
certain efficient algorithms (such as block QR) are not
applicable~\cite{gallopoulos2015parallelism}. Furthermore, for conditioning purposes, 
we have used the IMGS 
of Hoffman which cannot be written as a matrix-vector product
when the basis are sequentially known, thereby precluding efficient BLAS-2 routines~\cite{gallopoulos2015parallelism}.
Briefly, we offer three potential solutions. First, 
alternative orthogonalization algorithms, like the ``CMGSI" algorithm of Hoffman
may be better suited for parallelization~\cite{Hoffmann_IMGS}.
Second, specialized 
accelerator hardware may reduce the vector-vector product costs by offloading. Finally, one 
could consider alternative global pivot selection criteria by overlapping the pivot
search and orthogonalization computations. 

\section{Concluding remarks} \label{s:conc}

Dimensional and model-order reduction have a wide range of applications. In this paper, we have considered two of the most popular dimensional-reduction algorithms, SVD and QR decompositions, and summarized their most important properties. In most model-based dimensional reduction applications the model varies smoothly with parameter variation. In such cases, matrices like the Kahan one are rarely (if ever) encountered in practice. Instead, the approximation problem is characterized by a fast decaying Kolmogorov $n$-width. Due to the equivalence we showed between the RB-Greedy algorithm and a certain QR pivoting strategies, we argue that, for many cases, a QR-based model reduction approach is preferable to an SVD-based one. The QR decomposition is faster and more easily parallelized while providing comparable approximation errors. 

Finally, we have described a new, publicly-available QR-based model reduction code, 
greedycpp~\cite{greedycpp}. 
Our code is
based on a well-conditioned version of the MGS algorithm which overcomes the stability issues which plague ordinary GS while being 
straightforward to parallelize (as compared to Householder reflections or Givens rotations). This massively parallel code, developed with model reduction in mind, performed QR decomposition on matrices as large as $10,000$-by-$3,276,800$ on the supercomputers Comet and Blue Waters. Parts of this code have been used to accelerate gravitational wave inference 
problems~\cite{smith2016fast,canizares2015accelerated,abbott2017first,abbott2017gw170814}. 

\section*{Acknowledgments}

\noindent We acknowledge helpful discussions with
Yanlai Chen,
Howard Elman,
Chad Galley,
Frank Herrmann,
Alfa Heryudono,
Saul Teukolsky,
and Manuel Tiglio. 
We thank
Priscilla Canizares, 
Collin Capano, 
Peter Diener,
Jeroen Meidam,
Michael Purrer, 
Rory Smith,
and Ka Wa Tsang for careful error reporting, code improvements and testing of early versions of the greedycpp code. 
Michael Purrer and Rory Smith for interfaces to the gravitational waveform models implemented in LALSimulation. SEF was 
partially supported by NSF award PHY-1606654 and the Sherman Fairchild Foundation.
HA was partially supported by NSF grant DMS-1521590.
Computations were performed on NSF/NCSA Blue Waters under allocation PRAC
ACI-1440083, on the NSF XSEDE network under allocations TG-PHY100033
and TG-PHY990007, and on the Caltech compute cluster Zwicky (NSF
MRI-R$^2$ award no.\ PHY-0960291).

\bibliographystyle{siam}
\bibliography{references}

\end{document}